\documentclass[11pt]{article}
\usepackage{liyang}

\newcommand{\Del}{\mathbf{{Del}}}
\newcommand{\Indicator}[1]{\mathbf{1}\left[#1\right]}

\usepackage{dsfont}

\newcommand{\Hypernb}[3]{\mathsf{Hypernb}(#1, #2, #3)}

\newcommand{\tprod}{\prod}

\begin{document}

\title{Trace reconstruction from local statistical queries}

\author{
Xi Chen\thanks{Supported by NSF grants CCF-1703925, IIS-1838154, CCF-2106429 and CCF-2107187.} 
\\
Columbia University\\
xichen@cs.columbia.edu
\and
Anindya De\thanks{Supported by NSF grants CCF-1926872, CCF-1910534 and CCF-2045128.}
\\
University of Pennsylvania\\
anindyad@cis.upenn.edu
\and
Chin Ho Lee\thanks{Work done in part at Harvard University, supported by Madhu Sudan’s and Salil Vadhan’s Simons Investigator Awards while at Harvard University.}
\\
North Carolina State University\\
chinho.lee@ncsu.edu
\and
Rocco A. Servedio\thanks{Supported by NSF grants IIS-1838154, CCF-2106429, CCF-2211238 and by the Simons Collaboration on Algorithms and Geometry.} 
\\
Columbia University\\
rocco@cs.columbia.edu
}

\date{}
\maketitle

\thispagestyle{empty}

\begin{abstract}
The goal of \emph{trace reconstruction} is to reconstruct an unknown $n$-bit string $x$ given only independent random \emph{traces} of $x$, where a random trace of $x$ is obtained by passing $x$ through a deletion channel.
A \emph{Statistical Query} (SQ) algorithm for trace reconstruction is an algorithm which can only access statistical information about the distribution of random traces of $x$ rather than individual traces themselves.
Such an algorithm is said to be \emph{$\ell$-local} if each of its statistical queries corresponds to an $\ell$-junta function over some block of $\ell$ consecutive bits in the trace.  
Since several --- but not all --- known algorithms for trace reconstruction fall under the local statistical query paradigm, it is interesting to understand the abilities and limitations of local SQ algorithms for trace reconstruction.

In this paper we establish nearly-matching upper and lower bounds on local Statistical Query algorithms for both worst-case and average-case trace reconstruction.
For the worst-case problem, we show that there is an $\tilde{O}(n^{1/5})$-local SQ algorithm that makes all its queries with tolerance $\tau \geq 2^{-\tilde{O}(n^{1/5})}$, and also that any $\tilde{O}(n^{1/5})$-local SQ algorithm must make some query with tolerance $\tau \leq 2^{-\tilde{\Omega}(n^{1/5})}$.
For the average-case problem, we show that there is an $O(\log n)$-local SQ algorithm that makes all its queries with tolerance $\tau \geq 1/\poly(n)$, and also that any $O(\log n)$-local SQ algorithm must make some query with tolerance $\tau \leq 1/\poly(n).$
\end{abstract}

\newpage

\setcounter{page}{1}


\section{Introduction} \label{sec:intro}

In the \emph{trace reconstruction} problem, the goal is to reconstruct an unknown string $x \in \zo^n$ given access to independent random \emph{traces} of $x$, where a random trace of $x$ is a string obtained by passing $x$ through a deletion channel that independently deletes each bit with probability $\delta$ and concatenates the surviving bits.  
Trace reconstruction has been a well-studied problem since the early 2000s~\cite{Lev01a,Lev01b,BKKM04}, and some combinatorial variants of the problem were already considered in the 1970s~\cite{Kalashnik73}. 
Over the past decade, a wide range of algorithmic results and lower bounds have been established for many variants of the trace reconstruction problem, including worst-case \cite{MPV14,DOS17,NazarovPeres17,HL18,Chase19,Chase20}, average-case \cite{PeresZhai17,HPP18,HPPZ20,Rubinstein23}, and smoothed analysis \cite{CDLSS20smoove} versions, the low deletion rate regime \cite{CDLSS-itcs-21}, approximate trace reconstruction \cite{DRRS20,CDK21,ChasePeres21,CDLSS22}, coded trace reconstruction \cite{CGMR19_coded,BLS20_coded}, variants in which different bits of the source string have different deletion probabilities \cite{HHP18}, circular trace reconstruction \cite{NarayananRen20}, trace reconstruction on trees \cite{DRR19,BR23}, population recovery variants \cite{BCFSSfocs19,Narayanan20,Narayanan21a}, connections to other problems such as mixture distribution learning \cite{KMMP19}, and more \cite{GSZ20,SB21}.

The original, and arguably most fundamental, versions of the problem are the ``worst-case'' and ``average-case'' versions with constant deletion rate $\delta \in (0,1)$.  In the worst-case problem the source string $x$ is an arbitrary (worst-case) element of $\zo^n$, and in the average-case problem the source string $x$ is selected uniformly at random from $\zo^n$; equivalently, an average-case algorithm is only required to succeed for a $1-o_n(1)$ fraction of all $2^n$ possible source strings $x \in \zo^n$.  These two problems are the focus of our work, so in the rest of this paper we consider worst-case and average-case trace reconstruction and we always assume that the deletion rate $\delta$ is an arbitrary (known) constant in $(0,1)$.

Despite much effort, there are mildly exponential gaps between the best known upper bounds and lower bounds for both worst-case and average-case trace reconstruction. Improving on earlier $2^{\tilde{O}(n^{1/2})}$-trace   and $2^{\tilde{O}(n^{1/3})}$-trace algorithms of \cite{HMPW08,DOS17,NazarovPeres17}, in \cite{Chase20} Chase gave an algorithm for worst-case trace reconstruction that uses $2^{\tilde{O}(n^{1/5})}$ traces. The best known lower bound, also due to Chase \cite{Chase19}, is $\tilde{\Omega}(n^{3/2})$ traces (improving on earlier $\tilde{\Omega}(n^{5/4})$ and $\Omega(n)$ lower bounds \cite{HL18,BKKM04}).
For the average-case problem, improving on earlier $\exp(O((\log n)^{1/2}))$-trace and 
$\exp(O((\log n)^{1/3}))$-trace algorithms \cite{PeresZhai17,HPP18,HPPZ20}, Rubinstein \cite{Rubinstein23} recently gave an $\exp(\tilde{O}((\log n)^{1/5}))$-trace algorithm.
The best known average-case lower bound,~due to Chase \cite{Chase19}, is $\tilde{\Omega}((\log n)^{5/2})$ traces, improving on an earlier $\tilde{\Omega}((\log n)^{9/4})$ lower bound \cite{HL18}.

These substantial gaps naturally suggest the study of restricted classes of algorithms for trace reconstruction, with the hope that it may be possible to obtain sharper results.  This is the starting point of our work: we propose to study the trace reconstruction problem from the vantage point of \emph{statistical query} algorithms.  As our main contribution we obtain fairly sharp upper and lower bounds on \emph{local} statistical query algorithms for trace reconstruction, as described below.

\medskip

\noindent {\bf Statistical Query trace reconstruction algorithms.}
The \emph{Statistical Query} (SQ) model \cite{Kearns:98} was first introduced by Kearns as a means to obtain PAC learning algorithms that can tolerate random classification noise.  In the decades since then, the SQ model has emerged as a major topic of study in its own right in computational learning theory and related fields such as differential privacy and optimization. An attractive feature of the SQ model is that it is powerful enough to capture state-of-the-art algorithms in a variety of different settings, yet it is also amenable to proving \emph{unconditional} lower bounds.

SQ algorithms can only access data through noisy estimates of the expected values of user-generated query functions.  In the context of trace reconstruction, an SQ oracle takes as input a bounded \emph{query function} $q: \zo^n \to [-1,1]$ and a \emph{tolerance parameter} $\tau \in (0,1)$ that are provided by the reconstruction algorithm.
It returns a value $\widehat{P}_q$ which satisfies
$\abs{\widehat{P}_q - P_q} \leq \tau$, where $P_q$ is the expected value of $q$ on a random trace, i.e.~$P_q := \Ex_{\by \sim \Del_\delta(x)}[q(\by)].$\footnote{Since the length of each trace is at most $n$, we view each trace $\by$ as padded with a suffix of $n-|\by|$  zeros, so the argument to $q$ is actually $\by 0^{n - |\by|}$. This is equivalent to assuming that the $n$-bit source string $x$ is padded with an infinite suffix of $0$-bits.} Thus an SQ algorithm for trace reconstruction does not receive any actual traces of $x$; rather, it can only use aggregate statistical information about the overall  distribution of traces.  

To the best of our knowledge, the current paper is among the first works that explicitly considers the trace reconstruction problem from the perspective of statistical queries (see also \cite{CGLSZ23}, which we discuss in more detail below).  However, in hindsight the earliest nontrivial algorithms for worst-case trace reconstruction \cite{HMPW08,DOS17,NazarovPeres17} already made it evident that SQ algorithms --- in fact, SQ algorithms which use extremely simple query functions --- could be effective for trace reconstruction.  The algorithms of \cite{HMPW08,DOS17,NazarovPeres17} all work by using the traces from $\Del_\delta(x)$ only to obtain high-accuracy estimates of the $n$ values $\Ex_{\by \sim \Del_\delta(x)}[\by_i]$ for $i \in [n]$ and then doing some subsequent computation on those estimated values; thus they correspond to SQ algorithms in which each query function is simply a Boolean dictator function, i.e.~a 1-junta. On the other hand, the highly efficient average-case trace reconstruction algorithms of \cite{PeresZhai17,HPP18,HPPZ20,Rubinstein23}, which use a sub-polynomial number of traces, involve various ``alignment'' routines which attempt to identify locations in individual received traces that correspond to specific locations in the source string.  These algorithms seem to make essential use of individual traces and do not seem to be compatible with the SQ model.  So given that some, but not all, known trace reconstruction algorithms correspond to SQ algorithms, it is of interest to study both the abilities and limitations of SQ algorithms for trace reconstruction.

In this work we consider a natural class of SQ algorithms, which we call \emph{$\ell$-local} SQ algorithms. An $\ell$-local query function $q: \zo^n \to [-1,1]$ is an $\ell$-junta over some $\ell$ consecutive bits of its input string, i.e.~for all $y$, $q$ satisfies $q(y)=q'(y_i,y_{i+1},\dots,y_{i+\ell-1})$ for some index $i$ and some function $q': \zo^\ell \to [-1,1].$ We say that an  algorithm is an \emph{$\ell$-local SQ algorithm with tolerance $\tau_0$} if all of its calls to the SQ oracle are made with $\ell$-local query functions and the tolerance parameter for each call is at least $\tau_0$.

The results of \cite{DOS17,NazarovPeres17} already show that 1-local SQ algorithms with tolerance $\tau_0 = 2^{-\tilde{O}(n^{1/3})}$ can successfully perform worst-case trace reconstruction, and moreover \cite{DOS17,NazarovPeres17} additionally show that tolerance $\tau_0 = 2^{-\tilde{\Omega}(n^{1/3})}$ is required for any 1-local SQ worst-case trace reconstruction algorithm. Thus, in analyzing the abilities and limitations of $\ell$-local algorithms for trace reconstruction for a particular value of $\ell$, our  goal is to determine the tolerance which is necessary and sufficient for such algorithms to succeed in worst-case or average-case trace reconstruction.
A simple argument which we give in \Cref{sec:localSQ} shows that any $\ell$-local SQ algorithm (which may be adaptive) using tolerance $\tau_0$ can be converted to a nonadaptive SQ algorithm that makes at most $n 2^\ell$ queries, all of which are $\ell$-local ``subword'' queries (defined in \Cref{sec:localSQ}) of tolerance $\tau_0 2^{-\ell}$. Moreover, a standard argument shows that any nonadaptive SQ algorithm which  makes $M$ statistical queries, each with tolerance at least $\tau_0$,  can be simulated in the obvious way by a standard trace reconstruction algorithm that uses $\poly(\log M,1/\tau_0)$ independent traces from $\Del_\delta(x).$
Thus, we will be particularly interested in identifying the value $\ell$ of the locality parameter for which tolerance (roughly) $2^{-\ell}$ is both necessary and sufficient for trace reconstruction. As we explain next, our main results do precisely this, for both worst-case and average-case trace reconstruction.

\subsection{Our results} \label{sec:results}

We give upper and lower bounds on local SQ algorithms for both worst-case  and average-case trace reconstruction. 
Our upper and lower bounds match each other up to fairly small factors for both the worst-case and average-case versions of the problem.

\medskip

\noindent {\bf The worst-case problem.} Our main lower bound is the following result, which gives a lower bound on the tolerance for $n^{1/5}$-local SQ algorithms performing worst-case trace reconstruction:

\begin{theorem} [Worst-case lower bound, informal version of \Cref{thm:worstlower}]  \label{thm:informal-worstlower}
Fix any constant deletion rate $0<\delta<1$.  
For $\ell = \tilde{\Theta}(n^{1/5})$, any $\ell$-local SQ algorithm for worst-case trace reconstruction must have tolerance $\tau_0 = \exp(-\tilde{\Omega}(n^{1/5})).$
\end{theorem}

Our algorithmic result for the worst-case problem shows that this lower bound is essentially optimal:

\begin{theorem} [Worst-case upper bound, informal version of \Cref{thm:worstupper}]  \label{thm:informal-worstupper}
  Fix any constant deletion rate $0<\delta<1$.  There is a  $\tilde{O}(n^{1/5}))$-local SQ algorithm for the worst-case trace reconstruction problem with tolerance $\tau_0=\exp(-\tilde{O}(n^{1/5}))$.
\end{theorem}

\noindent {\bf The average-case problem.} As mentioned earlier, the state-of-the-art average-case trace reconstruction algorithms of \cite{PeresZhai17,HPP18,HPPZ20,Rubinstein23} do not seem to be compatible with the SQ model.
Recall that those algorithms use $2^{O((\log n)^c)}$ traces, for $c \in \{1/5, 1/3, 1/2\}$, and thus any SQ analogue of those algorithms would have tolerance $\approx 2^{-O((\log n)^{c})}$.  We show that no $O(\log n)$-local (or even $n^{0.49}$-local) SQ algorithm for average-case trace reconstruction can succeed with such a coarse tolerance parameter:

\begin{theorem} [Average-case lower bound, informal version of \Cref{thm:avglower}]  \label{thm:informal-avglower}
Fix any constant deletion rate $0<\delta<1$.  Any $\ell$-local SQ algorithm for average-case trace reconstruction must have tolerance $\tau_0 \leq \ell/\sqrt{n}.$
\end{theorem}

Finally, we give an average-case $O(\log n)$-local SQ algorithm that has inverse polynomial tolerance:

\begin{theorem} [Average-case upper bound, informal version of \Cref{thm:avgupper}]  \label{thm:informal-avgupper}
Fix any constant deletion rate $0<\delta<1$.  There is an $O(\log n)$-local SQ  algorithm for average-case trace reconstruction with tolerance $\tau_0=1/\poly(n)$.
\end{theorem}

Our results can be summarized as follows:  As discussed immediately before \Cref{sec:results}, we may say that an $\ell$-local SQ algorithm with tolerance $\tau_0$ has overall complexity $\poly(n2^\ell,1/\tau_0)$. 
\Cref{thm:informal-worstlower,thm:informal-worstupper} together say that the optimal complexity of worst-case local SQ trace reconstruction is $2^{\tilde{\Theta}(n^{1/5})}$, and \Cref{thm:informal-avglower,thm:informal-avgupper} together say that the optimal complexity of average-case local SQ trace reconstruction is $n^{\Theta(1)}.$

\subsection{Discussion and techniques}

\noindent {\bf The worst-case setting.}
\Cref{thm:informal-worstlower} and \Cref{thm:informal-worstupper} should be contrasted with recent results of Cheng et al.~\cite{CGLSZ23}, which consider a restricted class of local SQ algorithms known as \emph{$\ell$-mer based} algorithms. 
As defined by Mazooji and Shomorony~\cite{MS23}, the \emph{$\ell$-mer density map} is a certain vector of statistics about the frequency of length-$\ell$ subwords\footnote{Recall that a \emph{subword} of $x$ is a sequence of bits that occur consecutively in $x$, i.e.~$x_i x_{i+1} \cdots x_{i+\ell-1}$, whereas a \emph{substring} of $x$ is a subsequence of bits that need not occur consecutively, i.e. $x_{i_1} x_{i_2} \cdots x_{i_\ell}$.} of the source string $x \in \zo^n$.
\cite{MS23} gave an algorithm which, for constant deletion rate $0 < \delta < 1/2$,  constructs an $\eps$-accurate (in $\ell_\infty$ distance) estimate of the $\ell$-mer density map using $\poly(n,2^\ell,1/\eps)$ traces.  
Cheng et al.~\cite{CGLSZ23} defined a trace reconstruction algorithm to be \emph{$\ell$-mer based} if it only uses the $\ell$-mer density map of $x$, and observed that the algorithm of \cite{MS23} (see in particular Lemma~6 of \cite{MS23} and its proof) only uses local statistical information about traces, and hence is a local SQ algorithm.

The main result of Cheng et al.~is a proof that any $n^{1/5}$-mer based algorithm for worst-case trace reconstruction must have tolerance $\tau_0 = 2^{-\tilde{\Omega}(n^{1/5})}$.
Our \Cref{thm:informal-worstlower} generalizes this result  because it gives a lower bound for the entire class of $n^{1/5}$-local SQ algorithms, which includes the class of $n^{1/5}$-mer based algorithms by the results described above.
We remark that \Cref{thm:informal-worstlower} also has a shorter and simpler proof than Theorem~2 of \cite{CGLSZ23}.

At a high-level, we obtain \Cref{thm:informal-worstlower} by a reduction to proving a 1-local SQ lower bound on $n$-bit source strings that are  ``gappy.''
These are strings in which every two 1s are separated by $\gg n^{1/5}$ zeros (see \Cref{def:gappy} for the precise definition).
The intution here is that for a gappy string, any $n^{1/5}$-bit subword in its traces is very unlikely to contain two 1s, and so all the useful information is contained in subword queries of Hamming weight at most 1, which can then be further reduced to 1-bit queries.
Then we can adapt the lower bound arguments in \cite{DOS17} to gappy strings to obtain our lower bound.


Turning to \Cref{thm:informal-worstupper}, Cheng et al.~observed that 
the $2^{\tilde{O}(n^{1/5})}$-trace algorithm of \cite{Chase20} for worst-case trace reconstruction can be interpreted as a $\tilde{O}(n^{1/5})$-mer based algorithm with tolerance $\tau_0 = 2^{-\tilde{O}(n^{1/5})}$. By the earlier observation of Cheng et al.~mentioned in the first paragraph and the \cite{MS23} algorithm, which works provided that the deletion rate $\delta$ lies in $(0,1/2)$, this means that Chase's algorithm can be expressed as a $\tilde{O}(n^{1/5})$-local SQ algorithm which has tolerance $\tau_0 = 2^{-\tilde{O}(n^{1/5})}$ when $\delta \in (0,1/2).$
Our \Cref{thm:worstupper} is based on a similar observation about Chase's algorithm, but applied directly to the local SQ model without going through the notion of $k$-mer statistics. Our approach is based on techniques and arguments from \cite{CDLSS20smoove}; using these techniques allows our argument to apply more generally to the entire range of deletion rates $\delta \in (0,1)$.
%
%
%
%
%
\medskip

\noindent {\bf Average-case.}
The average-case lower bound of \Cref{thm:informal-avglower} is proved using a fairly simple argument based on ``hiding'' a bit which might be either 0 or 1 in the middle of the source string.  We turn to the average-case upper bound.

The average-case SQ algorithm described in \Cref{thm:informal-avgupper} is obtained by adapting an algorithm for \emph{smoothed} trace reconstruction to the SQ model. The \cite{CDLSS20smoove} paper gives an algorithm for ``smoothed'' trace reconstruction, which is a generalization of the average-case trace reconstruction problem.
While the algorithm of \cite{CDLSS20smoove} only interacts with the input traces by using them to form empirical estimates of subword frequencies in traces, it is not trivially an SQ algorithm.  This is because the \cite{CDLSS20smoove} algorithm estimates these subword frequencies across a range of different deletion probabilities $\delta, \delta + \Delta, \delta + 2 \Delta, \dots$ up to $(\delta + 1)/2.$
In the usual (non-SQ) trace reconstruction setting where traces are available, it is trivial to simulate access to $\Del_{\delta'}(x)$ given access to $\Del_\delta(x)$ for any $\delta' > \delta$, simply by drawing $\by \sim \Del_\delta(x)$ and deleting each bit of $\by$ independently with probability ${\frac {1-\delta'}{1-\delta}}$. But in the SQ setting, we only have access to statistical queries of traces drawn from $\Del_\delta(x)$ rather than individual traces.
We circumvent this issue by showing that any algorithm that makes $\ell$-local statistical queries with tolerance $\tau$ to $\Del_{\delta'}(x)$, for $\delta' > \delta$, can be simulated by an algorithm that makes only $\ell'$-local statistical queries with tolerance $\tau'$ to $\Del_{\delta}(x)$, where (roughly speaking) $\ell' \approx \ell/(1-\delta')$ and $\tau' = \Theta(\tau).$
With this ingredient in hand, the algorithm of \cite{CDLSS20smoove} is easily adapted to give \Cref{thm:informal-avgupper}.


%
%

\subsection{Future work}

Several natural questions suggest themselves for future work.  
Perhaps the foremost among these is the following:  Given \Cref{thm:informal-worstupper}, the current state-of-the-art unrestricted algorithm for the general worst-case trace reconstruction problem is an $\tilde{O}(n^{1/5})$-local SQ algorithm.  Might it be the case that this is in fact an optimal algorithm for trace reconstruction?  We currently seem quite far from being able to resolve this (recall that the state of the art in lower bounds for unrestricted worst-case trace reconstruction algorithms is only $\tilde{\Omega}(n^{3/2})$ traces \cite{Chase20}).

A partial step towards answering the above bold question would be to establish lower bounds on general SQ algorithms for worst-case trace reconstruction, i.e.~SQ algorithms that are not assumed to have bounded locality.  It is difficult to imagine how queries that depend on far-separated portions of an input trace could be useful, but proving this seems quite challenging.  

As a concrete first goal along these lines, a generalization of the notion of an $\ell$-local SQ is the notion of a \emph{size-$s$} SQ. A size-$s$ SQ is an SQ which asks for the expected value of some $s$-junta function $q'(\by_{i_1},\dots,\by_{i_s})$ of a random trace $\by$, but unlike an $\ell$-local  SQ the input bits of the junta do not need to form a consecutive block of positions in $\by$.  Similar to \Cref{lem:universal}, a size-$s$ SQ algorithm can be assumed  without loss of generality to use only query functions of the form $\Indicator{y_{i_1},\dots,y_{i_s}}=w$ as $(i_1,\dots,i_s)$ ranges over ${[n] \choose s}$ and $w$ ranges over $\zo^s$.
 Even the following goal appears to be  quite challenging:  
 \begin{quote}
 Show that any SQ algorithm for the worst-case trace reconstruction problem that makes only size-2 queries must have tolerance $\tau=1/n^{\omega(1)}.$
 \end{quote}
 We believe that this is an interesting target problem for future work.


\section{Preliminaries} \label{sec:prelims}

\noindent {\bf Notation.}
Given integers $a\le b$ we write $[a:b]$ to denote $\{a,\ldots,b\}$. 
It will be convenient for us to index a binary string~$x \in \zo^n$
  using $[0:n-1]$ as $x=(x_0,\dots,x_{n-1})$. 
We write $\ln$ to denote natural logarithm and $\log$ to denote logarithm to the base 2.
We write $|x|$ to denote the length of a string $x$.

We denote the set of non-negative integers by $\Z_{\geq 0}$.
We write $D_r(z)$ to denote the closed disk in the complex plane of radius $r$ centered at $z \in \C$, and $\partial D_r(z)$ to denote the circle which is the boundary of that disk.
%

\medskip

\noindent {\bf Subwords.} Fix a string $x \in \zo^n$ and an integer $k \in [n]$. A $k$-subword of $x$ is a 
 (contiguous) subword of $x$ of length $k$, given by $(x_a, x_{a+1}, \dots, x_{a+k-1})$ for some $a \in [0: n-k]$. 
 Given such a string $x$ and integers $0 \leq a < b \leq n-1$, we write $x[a:b]$ to denote the subword $(x_a, x_{a+1}, \dots, x_b)$.
 For a string $w \in \zo^k$, let $\#(w,x)$ denote the number of occurrences of $w$ as a subword of $x$. 

\medskip

\noindent {\bf Distributions.}
We use bold font letters to denote probability distributions and 
  random variables, which should be clear from the context.
We write ``$\bx \sim \bX$'' to indicate that random variable~$\bx$~is 
  distributed according to distribution $\bX$.

\medskip

\noindent {\bf Deletion channel and traces.}
Throughout this paper the parameter $\delta:0 <$ $\delta < 1$ denotes~the \emph{deletion probability}, and we write $\rho$ to denote the retention probability $\rho=1-\delta$.  Given a string $x \in \zo^n$, we write $\Del_\delta(x)$ to denote the distribution of the string that results from passing  $x$ through the $\delta$-deletion channel (so the distribution $\Del_\delta(x)$ is supported on $\zo^{\leq n}$), and we refer to a string drawn from $\Del_\delta(x)$ as a \emph{trace} of $x$.  Recall that a random trace $\by \sim \Del_\delta(x)$ is obtained by independently deleting each bit of $x$ with probability $\delta$ and concatenating the surviving bits.\hspace{0.05cm}\footnote{For simplicity in this work we assume that the deletion probability $\delta$ is known to the reconstruction algorithm.}

For $x \in \zo^n$ (recall that we index the bits of $x$ as $(x_0,\dots,x_{n-1}$)) we view a draw of a trace $\by \sim \Del(x)$ as corresponding to a $\rho$-biased random draw of a subset $\bR \subseteq [0:n-1]$, where the elements of $\bR$ are the bits that are \emph{retained} in $x$ to obtain $\by$.  So if the sorted elements of $\bR$ are $\bR=\{\br_0 < \br_1 < \cdots < \br_{\boldm-1}\}$ for some $\boldm \leq n$, then the bits of the trace $\by=(\by_0,\by_1,\dots,\by_{\boldm-1})$ are $\by_0=x_{\br_0}, \by_1 = x_{\br_1},$ and so on.  
\medskip

\subsection{Local Statistical Query algorithms}
\label{sec:localSQ}

As described earlier, an $\ell$-local query function $q: \zo^n \to [-1,1]$ is a function
\[
q(y)=q'(y_i,y_{i+1},\dots,y_{i+\ell-1})
\]
for some index $i$ and some function $q': \zo^\ell \to [-1,1],$ i.e.~a real-valued bounded $\ell$-junta over consecutive input variables.  
An algorithm is an \emph{$\ell$-local SQ algorithm with tolerance $\tau_0$} if all of its calls to the SQ oracle are made with $\ell$-local query functions and the tolerance parameter for each call is at least $\tau_0$.

Let us say that an $\ell$-local query function is a \emph{subword query} if it is of the form
\begin{equation}
\label{eq:subword}
q'(y)= \Indicator{(y_i,\dots,y_{i+\ell-1})=w}
\end{equation}
for some string $w \in \zo^\ell$.
The following simple lemma shows that without loss of generality, every $\ell$-local SQ algorithm makes at most $n 2^\ell$ (non-adaptive) queries, corresponding to all possible length-$\ell$ subword queries:

\begin{lemma} \label{lem:universal}
Let $A$ be an $\ell$-local SQ algorithm with tolerance $\tau_0$ (note that $A$ may make any number of calls to the SQ oracle and may be adaptive, i.e.~the choice of later queries may depend on the responses received on earlier queries). 
Then there is an algorithm $A'$ with the same behavior as $A$ which makes $n 2^\ell$ queries (all possible length-$\ell$ subword queries), each with tolerance $\tau_0/2^\ell.$
\end{lemma}

\begin{proof}
The algorithm $A'$ makes all $n 2^\ell$ subword queries of the form given in \Cref{eq:subword}, where $i$ ranges over $[n]$ and $w$ ranges over $\zo^\ell$.  It makes each such subword query with tolerance parameter $\tau_0/2^\ell$.
Let $p_{i,w} = \Prx_{\by \sim \Del(x)}[(\by_i,\dots,\by_{i+\ell-1})=w]$ and let $\widehat{p}_{i,w}$ be the value received from the SQ oracle in response to the query \eqref{eq:subword}, so $|\widehat{p}_{i,w} - p_{i,w}| \leq \tau_0/2^\ell$.

Let $(q,\tau_0)$ be any (query function, tolerance) pair that $A$ may make in the course of its execution. We show that a $\pm \tau_0$-accurate estimate $\widehat{P}_q$ of $P_q$ can be computed from the responses to the $n 2^\ell$ queries of $A'$.
This is easily seen to imply the lemma.

Since $A$ is $\ell$-local, the expected value $P_q$ is
\[
P_q = \Ex_{\by \sim \Del_\delta(x)}[q'(\by_i,\dots,\by_{i+\ell-1})]
\]
for some $q': \zo^\ell \to [-1,1]$ and some $i \in [n]$. Since
\[
P_q = \sum_{w \in \zo^\ell} p_{i,w} \cdot q'(w),
\]
by setting $\widehat{P}_q$ to be
\[
P_q = \sum_{w \in \zo^\ell} \widehat{p}_{i,w} \cdot q'(w),
\]
recalling that $|q'(w)| \leq 1$ for all $w$, the triangle inequality gives
\[
|\widehat{P}_q-P_q|
=
\left | \sum_{w \in \zo^\ell} (\widehat{p}_{i,w} - p_{i,w}) \cdot q'(w)
\right |
\leq \max_{w} |q'(w)| \cdot 
\sum_{w} |\widehat{p}_{i,w} - p_{i,w}|
\leq \sum_{w} \tau_0/2^\ell \leq \tau_0
\]
as desired.
\end{proof}


\section{Worst-case lower bounds} \label{sec:worst-lower}

In this section we prove the following lower bound on local SQ algorithms for the worst-case trace reconstruction problem:

\begin{theorem} [Worst-case lower bound]  \label{thm:worstlower}
Fix any constant deletion rate $0<\delta<1$.  For a suitable absolute constant $c_0$, any $c_0n^{1/5}/(\log n)^{2/5}$-local SQ algorithm for worst-case trace reconstruction must have tolerance $\tau_0 < \exp(-\Omega(n^{1/5}/(\log n)^{2/5})).$
\end{theorem}

\noindent {\bf Setup.}
Fix any $0<\delta<1$.
For notational clarity let us write $\ell := c_0n^{1/5}/(\log n)^{2/5}$.
Given an $n$-bit source string $x$, an index $i \in [0:n-1]$, and an $\ell$-bit string $w$, we define the value
\begin{equation} \label{def:pxiw}
p_{x,i,w} := \Prx_{\by \sim \Del_\delta(x)}[(\by_i,\dots,\by_{i+\ell-1})=w],
\end{equation}
so $p_{x,i,w}$ is the probability that a random trace of $x$ has $w$ as the subword starting in position $i$. 
We refer to the vector $(p_{x,i,w})_{i \in [0:n-1],w \in \zo^\ell}$ as the \emph{$\ell$-subword signature} of $x$.

We will prove the following:
\begin{lemma} \label{lem:twodistinct}
  For a suitable absolute constant $c_0$, there are distinct $n$-bit strings $a \neq a' \in \zo^n$ whose $\ell$-subword signatures are very close to each other in $\ell_\infty$-distance: more precisely,
\begin{equation} \label{eq:close}
\text{For all $i \in [0:n-1], w \in \zo^\ell$, we have~}
|p_{a,i,w} - p_{a',i,w}| \leq 
\exp(-2c_0n^{1/5}/(\log n)^{2/5}).
\end{equation}
\end{lemma}

To see why \Cref{lem:twodistinct} implies \Cref{thm:worstlower}, let $A$ be any $\ell$-local SQ algorithm with tolerance $\exp(-c_0 n^{1/5}/(\log n)^{2/5}).$
By \Cref{lem:universal}, there is an algorithm $A'$ with the same behavior as $A$ which makes only subword queries for subwords of length $\ell$, where each query of $A'$ has tolerance $\exp(-c_0 n^{1/5}/(\log n)^{2/5})/2^\ell>\exp(-2c_0 n^{1/5}/(\log n)^{2/5})$.
By \Cref{eq:close}, a query for the value of $p_{x,i,w}$ can be answered with the value
$
q_{i,w} =
{\frac {p_{a,i,w} + p_{a',i,w}} 2}
$
whether the source string $x$ is $a$ or $a'$. But this means that it is impossible for $A$ to be an algorithm which successfully solves the worst-case trace reconstruction problem.

In the rest of this section, we focus on establishing \Cref{lem:twodistinct}.

We require the following simple definition:
\begin{definition} \label{def:gappy}
Given $t > 1$, we say that a string $x \in \zo^n$ is \emph{$t$-gappy} if it is of the form
\[
x = b_0 0^{t-1} b_1 0^{t-1}\cdots b_{n/t - 1} 0^{t-1}
\]
for some string $b_0,b_1,\dots,b_{n/t - 1} \in \zo^{n/t}$.
\end{definition}
{Recall $\rho = 1-\delta$.
Fix 
\begin{equation} \label{eq:t}
t := {\frac {100 \log(n) \ell} \rho} = \Theta(n^{1/5} (\log n)^{3/5}).
\end{equation}
The two strings $a,a'$ whose existence is asserted by \Cref{lem:twodistinct} will both be $t$-gappy.  (We note that the argument of Cheng et al.~\cite{CGLSZ23} also used gappy strings.)

One reason that gappy strings are useful for us because they make it very easy to handle almost all of the $\ell$-bit strings $w \in \zo^\ell$ that we need to consider in order to establish \Cref{lem:twodistinct}. To see this, observe that any string $w$ containing at least two ones is very unlikely to be a length-$\ell$ subword of a random trace $\by$: since the source string is $t$-gappy, we expect consecutive ones in a random trace $\by$ to be at least $\rho t \gg \ell$ positions apart from each other.  More precisely, we have the following lemma:

\begin{lemma} \label{lem:handleheavy}
Let $x \in \zo^n$ be any $t$-gappy string, and let $w \in \zo^\ell$ be any string with at least two ones.  Then for any $i \in [0:n-1]$ we have
$p_{x,i,w} \leq 1/n^{49\ell}.$
\end{lemma}
\begin{proof}
Fix $0 \leq \alpha <\beta \leq \ell-1$ to be any two positions in $w$ such that $w_\alpha=w_\beta=1$, and let $\by \sim \Del_\delta(x)$. 
Observe that we have
$p_{x,i,w} \leq \Pr[\by_{i+\alpha}=\by_{i+\beta}=1]$.

Let $\bR=\{\br_0 < \br_2 < \cdots < \br_{\boldm-1}\} \subseteq [0:n-1]$ be the $\rho$-biased random subset of $[0:n-1]$ consisting of the indices that are retained in  $x$ to obtain $\by$.  
We may view the draw of $\bR$ as being carried out sequentially in independent stages $0,1,\dots$, where in each stage $s$ the element $s$ is included in $\bR$ with probability $\rho$.
Fix any outcome of stages $0,1,\dots$ up until $r_{i+\alpha}$ has been included in $\bR$. Even supposing that $x_{r_{i+\alpha}}=1$ (so that $y_{i+\alpha}=1$),  the probability that $x_{\br_{i+\beta}}=1$ (which equals $\Pr[\by_{i+\beta}=1]$) is at most (writing $k$ for $\beta-\alpha$)
\begin{align}
  \MoveEqLeft
\sum_{j \geq 1} \Pr[\text{exactly $k$ of the $jt$ indices in $r_{i+\alpha}+1,\dots,r_{i+\alpha} + jt$ are retained}] \label{eq:prob}\\
&= \sum_{j \geq 1} {jt \choose k} \rho^k \delta^{jt-k}\nonumber \\
&\leq \sum_{j \geq 1} (\rho jt)^k \cdot \delta^{jt/2}
= (\rho t)^k \sum_{j \geq 1} j^k \delta^{jt/2}
\tag{since $k \leq jt/2$}\nonumber\\
&\leq (\rho t)^\ell \sum_{j \geq 1} j^\ell (1-\rho)^{jt/2}
\tag{using $k \leq \ell$}\nonumber\\
&\leq (100 \log(n) \ell)^\ell \sum_{j \geq 1} j^\ell e^{-50 \log(n)\ell j}
\tag{by the choice of $t$, and using $(1-\rho)^{1/\rho} \leq e^{-1}$}. \nonumber
\end{align}
When $j=1$ the first term of the sum $\sum_{j \geq 1} j^\ell e^{-50 \log(n) \ell j}$ is $e^{-50 \log(n) \ell}$.
The ratio of successive terms of the sum is
\[
{\frac {(j+1)^\ell e^{-50 \log(n) \ell (j+1)}}
{j^\ell e^{-50 \log(n) \ell j}}} \leq 2^\ell e^{-50\log(n)\ell} = (2/n^{50})^\ell \ll 1/2.
\]
So the sum $\sum_{j \geq 1} j^\ell e^{-50 \log(n) \ell j}$ is at most $2e^{-50 \log(n) \ell}=2/n^{50\ell}$, and since $100 \log(n) \ell<n$ for $n$ sufficiently large, we get that $\eqref{eq:prob} \leq 1/n^{49\ell}.$ It follows that  $p_{x,i,w} \leq \Pr[\by_{i+\alpha}=\by_{i+\beta}=1] \leq 1/n^{49\ell}$ as claimed.
\end{proof}

Given \Cref{lem:handleheavy} it remains to argue about the $\ell+1$ strings $w \in \zo^\ell$ of Hamming weight 0 or 1. We handle the weight-1 strings by reducing their analysis to the analysis of \emph{one-bit} strings as follows: fix any $\alpha \in [0:\ell-1]$ and let $w=e_\alpha \in \zo^\ell$ be the string with a single 1 coordinate in position $\alpha$.
The following lemma, which we prove using \Cref{lem:handleheavy}, shows that for any gappy source string $x$ the value of $p_{x,i,e_{\alpha}}$ is very  close to the expected value of a \emph{single} location in a random trace.  (A sharper bound could be obtained with a bit more work, but the bound given by \Cref{lem:weightone} is sufficient for our purposes.)

\begin{lemma} \label{lem:weightone}
Let $x \in \zo^n$ be any $t$-gappy string, and let $w = e_\alpha \in \zo^\ell$ be the string containing a single 1 in coordinate $\alpha$.  Then for any $i \in [0:n-1]$ we have
\[
\left|
\Prx_{\by \sim \Del_\delta(x)}[\by_{i+\alpha}=1]
-
p_{x,i,e_{\alpha}} 
\right|
\leq 2^{\ell-1} / n^{49 \ell}.
\]
\end{lemma}
\begin{proof}
We have
\[
\Prx_{\by \sim \Del_\delta(x)}[\by_{i+\alpha}=1] =
\sum_{w \in \zo^\ell: w_\alpha=1}
p_{x,i,w}, \quad \text{so}
\]
\[
0 \leq 
\Prx_{\by \sim \Del_\delta(x)}[\by_{i+\alpha}=1] - 
p_{x,i,e_{\alpha}}  =
\sum_{w \in \zo^\ell: w_\alpha=1, |w| \geq 2}
p_{x,i,w}
\leq (2^{\ell-1} - 1) /n^{49 \ell}
\]
where the inequality is \Cref{lem:handleheavy}.
\end{proof}

The one remaining $\ell$-bit string to consider is $w = 0^\ell$.  However, if all $2^\ell-1$ other strings have been handled successfully then this string is automatically handled as well:

\begin{lemma} \label{lem:laststring}
Fix $a,a' \in \zo^n$ and $i \in [0:n-1]$. Suppose that
for all $w \in \zo^\ell \setminus \{0^\ell\}$ we have
$|p_{a,i,w} - p_{a',i,w}| \leq \kappa$. Then $|p_{a,i,0^\ell} - p_{a',i,0^\ell}| \leq (2^\ell - 1)\kappa.$
\end{lemma}
\begin{proof}
This is an immediate consequence of $\sum_{w \in \zo^\ell} p_{x,i,w} = 1$, which holds for every $x$ and  $i$.
\end{proof}

Thus, it suffices to construct two $t$-gappy strings $a,a'$ whose one-bit statistics are very close:

\begin{lemma} \label{lem:onebitenough}
For $x \in \zo^n$ and $i \in [0:n-1]$ define 
\begin{equation} \label{eq:pxi}
p_{x,i} := \Prx_{\by \sim \Del_\delta(x)}[\by_i=1].
\end{equation}
Suppose that $a \neq a' \in \zo^n$ are two $t$-gappy strings such that for each $i \in [0:n-1]$ we have $|p_{a,i}-p_{a',i}| \leq \exp(-\Omega(n^{1/5}/(\log n)^{2/5})).$  Then for all $i \in [0:n-1]$, $w \in \zo^\ell$ we have 
\[
|p_{a,i,w} - p_{a',i,w}| \leq 2^\ell \cdot \exp(-\Omega(n^{1/5}/(\log n)^{2/5})) + 4^\ell / n^{49 \ell}
\leq \exp(-2c_0n^{1/5}/(\log n)^{2/5}).
\]
\end{lemma}
\begin{proof}
\Cref{lem:handleheavy} gives $|p_{a,i,w} - p_{a',i,w}| \leq  1/n^{49\ell}$ for $|w| \geq 2$.
\Cref{lem:weightone} and the assumption on $|p_{a,i}-p_{a',i}|$ gives $|p_{a,i,w}-p_{a',i,w}| \leq \exp(-\Omega(n^{1/5}/(\log n)^{2/5})) + 2^\ell/n^{49\ell}$ for $|w| = 1.$
Given these bounds, \Cref{lem:laststring} gives $|p_{a,i,0^\ell} - p_{a',i,0^\ell}| \leq 2^\ell \cdot \exp(-\Omega(n^{1/5}/(\log n)^{2/5})) + 4^\ell / n^{49\ell}.$
\end{proof}

\subsection{Establishing closeness of one-bit statistics}

Let us write $p_x = (p_{x,0},\dots,p_{x,n-1})$ to denote the $n$-dimensional vector in $[0,1]^n$ whose coordinates are given by \Cref{eq:pxi}.
From the results in the previous subsection it suffice to prove the following:

\begin{lemma} \label{lem:onebitclose}
There are two distinct $t$-gappy strings $a,a' \in \zo^n$ such that
for all $i \in [0:n-1]$ we have $\|p_a - p_{a'}\|_\infty \leq \exp(-\Omega(n^{1/5}/(\log n)^{2/5})).$
\end{lemma}

This is very similar to the main lower bound statement that was established in the two works \cite{DOS17,NazarovPeres17} (independently of each other); those papers considered ``one-bit statistics'' which correspond precisely to our $p_{x,i}$ quantities, and showed that there are two distinct strings $x,x' \in \zo^n$ (not restricted to be gappy) such that $|p_{x,i} - p_{x',i}| \leq \exp(-\Omega(n^{1/3}))$ for all $i \in [0:n-1]$.
In what follows we adapt their techniques to deal with $t$-gappy source strings.

Following \cite{DOS17}, given a pair of source strings $a,a' \in \zo^n$ we define the corresponding \emph{deletion-channel polynomial} (over $\C$) to be
\begin{equation} \label{eq:dcp}
P_{a,a'}(z) := \sum_{i=0}^{n-1} (p_{a,i} - p_{a',i}) \cdot z^i.
\end{equation}
We have
\begin{equation} \label{eq:a}
\|p_a - p_{a'}\|_\infty \leq  \|p_a - p_{a'}\|_1 \leq \sqrt{n} \max_{z \in \partial D_1(0)} |P_{a,a'}(z)|,
\end{equation}
where the second inequality is by Proposition~3.5 of \cite{DOS17} (the proof is a simple and standard computation about complex polynomials).
Thus our goal is to establish the existence of two distinct $t$-gappy strings $a \neq a' \in \zo^n$ for which 
$\max_{z \in \partial D_1(0)} |P_{a,a'}(z)|$ is small.
To do this, we begin by observing that since bit $j$ of a source string ends up in location $i$ of a trace with probability ${j \choose i} \rho^{i+1}\delta^{j-i}$, we have
\[
p_{a,i} =  \Prx_{\by \sim \Del_\delta(a)}[\by_i=1]
= \sum_{j=0}^{n-1} {j \choose i} \rho^{i+1} \delta^{j-i} a_j,
\quad \text{and hence} \quad
p_{a,i} - p_{a',i} = \sum_{j=0}^{n-1} (a_j-a'_j){j \choose i} \rho^{i+1} \delta^{j-i}.
\]
Hence (following \cite{DOS17,NazarovPeres17}) we get
\begin{align}
  P_{a,a'}(z) &= 
\sum_{i=0}^{n-1} \left(\sum_{j=0}^{n-1} (a_j - a'_j) {j \choose i} \rho^{i+1} \delta^{j-i}\right)z^i 
=
\rho \sum_{j=0}^{n-1} (a_j - a'_j) \delta^j \sum_{i=0}^{n-1} {j \choose i} \left( {\frac {\rho z} \delta}\right)^i\nonumber \\
&=
\rho \sum_{j=0}^{n-1} (a_j - a'_j) w^j   \quad \text{(taking $w=1-\rho + \rho z$)} \label{eq:wz}
\end{align}
where the last line used the binomial theorem and $\delta=1-\rho$. Now, let us write the $t$-gappy strings $a,a'$ as
\begin{equation}
a := b_0 0^{t-1}b_1 0^{t-1} \cdots b_{n/t} 0^{t-1},
\quad \quad \quad \quad
a' := b'_0 0^{t-1}b'_1 0^{t-1} \cdots b'_{n/t} 0^{t-1}
\label{eq:aaprime}
\end{equation}
for some $b,b' \in \zo^{n/t}$. From \Cref{eq:wz} we get that
\begin{equation} \label{eq:delicious}
P_{a,a'}(z) =  \rho \cdot \sum_{j=0}^{n/t-1} (b_j - b'_j) w^{jt}
\end{equation}
(the structure afforded by \cref{eq:delicious} is another reason why $t$-gappy strings are useful for us).  
Since $0<\rho=1-\delta<1$ is a constant, recalling \cref{eq:a} our goal is to establish the existence of a string $0^{n/t} \neq v = (v_0,\dots,v_{n/t-1}) \in \{-1,0,1\}^{n/t}$ such that 
\begin{equation} \label{eq:v}
\max_{\theta \in (-\pi,\pi]} \left|\sum_{j=0}^{n/t-1} v_j \left((1-\rho + \rho e^{i \theta})^t\right)^j \right|
\end{equation}
is small.

As described in Theorem~6.2 of \cite{DOS17}, a result of Borwein and Erd\'elyi~\cite{BE97} 
(specifically, the first proof of Theorem~3.3 in the ``special case'' on p. 11 of \cite{BE97}) establishes the following:

\begin{theorem} [\cite{BE97}] \label{thm:BE}
There are universal constants $c_1, c_2, c_3 > 0$ such that the following holds: For all $0 < a \leq c_1$ there exists an integer $2 \leq k \leq c_2/a^2$ and a nonzero vector $u \in \{-1,0,1\}^{k+1}$ such that $\max_{w \in D_{6a}(1)} |\sum_{j=0}^{k} u_j w^j | \leq \exp(-c_3/a)$.
\end{theorem}

Let $m = n^{1/5}/(\log n)^{2/5} = 1/a$, so $a = 1/m = (\log n)^{2/5}/n^{-1/5}.$
Recalling \cref{eq:t}, we have that $c_2/a^2 = c_2m^2 \ll n/(2t)$, so we get that there exists a vector $0^{n/(2t)} \neq u \in \{-1,0,1\}^{n/(2t)}$ such that
\begin{equation} \label{eq:boundy}
\max_{w \in D_{6/m}(1)} \left|
\sum_{j=0}^{n/(2t)-1} u_j w^j
\right| \leq \exp(-c_3 m).
\end{equation}
Routine geometry shows that if $|\theta| \leq {\frac 1 {mt}}$ then
$|1 - (1-\rho + \rho e^{i \theta})^t| \leq 6/m$, so we get that
\begin{equation} \label{eq:boundybound}
\max_{|\theta| \leq 1/(mt)} \left|
\sum_{j=0}^{n/(2t)-1} u_j \left((1-\rho + \rho e^{i \theta})^t\right)^j
\right| \leq \exp(-c_3 m).
\end{equation}



Now we can describe our final desired string $v \in \{-1,0,1\}^{n/t}$: it is obtained by padding $u$ with a prefix of $n/(2t)$ many zeros. We thus have
\begin{equation}
\label{eq:AB}
\eqref{eq:v} = \max_{\theta \in (-\pi,\pi]} \sum_{j=0}^{n/t - 1} v_j \left((1 - \rho + \rho e^{i\theta})^t\right)^j = 
\max_{\theta \in (-\pi,\pi]} \left(
\overbrace{(1 - \rho + \rho e^{i \theta})^{n/2}}^A \cdot
\overbrace{\sum_{j=0}^{n/(2t)-1} u_j \left( (1-\rho + \rho e^{i \theta})^t \right)^j}^B
\right).
\end{equation}
Since $|1-\rho + \rho e^{i \theta}| \leq 1$ for all $\theta \in (-\pi,\pi]$, we have that $|A|$ is always at most 1 and $|B|$ is always at most $n/(2t)$.  We bound \cref{eq:AB} by considering two possible ranges for $|\theta|$.
If $|\theta| \leq 1/(mt)$, then since $|A| \leq 1$, from \cref{eq:boundybound} we have that 
$
\eqref{eq:AB}
\leq 1 \cdot  |B| \leq \exp(-c_3 m).
$
On the other hand, if $|\theta|>1/(mt)$ then since $|B| \leq n/(2t)$ and $\rho$ is a constant between 0 and 1, we get that
$|1-\rho+\rho e^{i \theta}| \leq 1 - {\frac {c_\rho}{(mt)^2}}$, and hence
\[
\eqref{eq:AB}
\leq {\frac n {2t}} \cdot |A| \leq {\frac n {2t}} \cdot \left(1- \frac {c_\rho}{(mt)^2}\right)^{n/2}
\leq \exp(-c'_\rho n/(mt)^2)
\]
for two constants $c_\rho,c'_\rho>0$ that depend only on $\rho.$ 
Since $m = n^{1/5}/(\log n)^{2/5}$ and $n/(mt)^2 = \Theta(n^{1/5}/(\log n)^{2/5})$, for all $\theta \in (-\pi,\pi]$ we have that $\eqref{eq:v} \leq \exp(-\Omega(n^{1/5}/(\log n)^{2/5}))$,
so the proof of \Cref{lem:onebitclose} and hence of \Cref{thm:worstlower} is complete.

\section{Worst-case upper bounds} \label{sec:worst-upper}

In this section we will give a local SQ algorithm for worst-case trace reconstruction, proving \Cref{thm:informal-worstupper}.

\begin{theorem} [Worst-case upper bound]  \label{thm:worstupper}
Fix any constant deletion rate $0 < \delta < 1$.
There is a worst-case SQ trace reconstruction algorithm that makes only $(O(n^{1/5}\log^5 n))$-local queries with tolerance $\tau=2^{-O(n^{1/5}\log^5 n)}$.
\end{theorem}

\noindent {\bf Overview.}
As discussed in the introduction, \cite{CGLSZ23} showed that the state-of-the-art worst-case trace reconstruction algorithm of Chase~\cite{Chase20} can be interpreted as a $\tilde O(n^{1/5})$-mer based algorithm, and further observed that the work \cite{MS23} implicitly showed that for deletion rate $\delta < 1/2$, any $k$-mer based algorithm only relies on local statistics of random traces.
The same observation can also be inferred from the work \cite{CDLSS20smoove};  more generally, that work implicitly showed that for \emph{any} deletion rate $0<\delta<1$ (not just $\delta < 1/2$), Chase's algorithm can be interpreted as a local SQ algorithm.
We obtain \Cref{thm:worstupper} by making this interpretation explicit, without going through the notion of $k$-mer statistics.

In the case of $\delta < 1/2$, the observation in \cite{CDLSS20smoove,MS23} is the following.
Chase's algorithm is based on estimating (from below) a certain univariate polynomial $Q_x(z_0)$ 
at some point $z_0$ inside the shifted complex disc $D := \{\frac{z-\delta}{1-\delta}: \abs{z} \le 1\}$.
Moreover, the degree-$\ell$ coefficient of $Q_x$
can be estimated using $\ell$-local statistics.
When $\delta$ is bounded away from $1/2$, these works observed that the magnitude of the degree-$\ell$ term of $Q_x$ decays exponentially in $\ell$, and so the contribution from the high-degree terms is negligible and can be truncated from the evaluation.

In the case of $\delta \ge 1/2$, a point in $D$ can have magnitude $1$ or more, and so the high-degree terms in $Q_x$ need not decay in magnitude.
Instead of evaluating the polynomial on some point in $D$, \cite{CDLSS20smoove} applies a result by Borwein, Erd\'elyi, and K\'os~\cite{BEK99} (see \Cref{lem:BEK} below) which shows that there exists a value $t_0$ in the real interval $[\delta, \frac{1}{4} + \frac{3}{4}\delta]$ such that $Q_x(t_0)$ is almost as large as $Q_x(z_0)$, and as a result, we can estimate the truncation of $Q_x(t_0)$ instead.

\medskip

We now proceed to a detailed proof of \Cref{thm:worstupper}.
Let $\ell := 2n^{1/5}$.
Our $(O(n^{1/5}\log^5 n))$-local SQ algorithm in \Cref{thm:worstupper} is based on the following two lemmas.
(Throughout this section, it will be more convenient for us to phrase various quantities in terms of the retention rate $\rho = 1-\delta$.)
For a source string $x \in \zo^n$ and an $\ell$-bit pattern $w \in \zo^\ell$, let $P_{x,w}(z,t)$ be the following bivariate polynomial: 
\[
  P_{x,w}(z,t)
  := \sum_{0 \le i_1 < \cdots < i_\ell \le n-1} \prod_{k=1}^\ell \Indicator{x_{i_k} = w_k} z^{i_1}  \cdot t^{i_\ell - i_1 - (\ell-1)} .
\]

\begin{lemma} \label{lem:distinguish-src-poly}
  For every deletion rate $\delta \in (0,1)$, there is a constant $C_\rho$ such that the following holds.
  For every distinct pair of source strings $x,x' \in \zo^n$, there is a pattern $w \in \zo^\ell$, a point $z_0 \in \{ e^{i\theta}: \abs{\theta} \le n^{-2/5}\} \cup [1-\rho, 1-\frac{3}{4}\rho]$, and a real value $t_0 \in [1-\rho,1-\frac{3}{4}\rho]$, such that
  \[
    \abs{ P_{x,w}(z_0,t_0) - P_{x',w}(z_0,t_0) }
    \ge \exp\bigl(-C_\rho n^{1/5} \log^5 n\bigr) .
  \]
\end{lemma}

\begin{lemma} \label{lem:SQ-alg-for-src-poly}
  For every deletion rate $\delta \in (0,1)$, there exists an SQ algorithm that makes $C_\rho n^{1/5} \log^5 n$-local queries with tolerance $\exp(-C_\rho n^{1/5} \log^5 n)$ such that for every $w \in \zo^\ell$, $z \in \{ e^{i\theta}: \abs{\theta} \le n^{-2/5}\} \cup [1-\rho, 1-\frac{3}{4}\rho]$, and $t \in [1-\rho,1-\frac{3}{4}\rho]$ it outputs an estimate $\wh{P_{x,w}}(z,t)$ of $P_{x,w}(z,t)$ that is accurate to within $\pm 0.1 \cdot \exp(-C_\rho n^{1/5} \log^5 n)$.
\end{lemma}

\paragraph{Our $\ell$-local SQ algorithm (Proof of \Cref{thm:worstupper} assuming \Cref{lem:distinguish-src-poly,lem:SQ-alg-for-src-poly}).}
Given an unknown source string $x \in \zo^n$, our reconstruction algorithm enumerates every pair of distinct strings $x_1 \neq x_2 \in \zo^n$. For each such pair, it considers the triple $(w, z_0, t_0)$ for that pair whose existence is given by \Cref{lem:distinguish-src-poly}.
(Hence there are at most $2^{2n}$ many such triples $(w,z_0,t_0)$ considered in total.)
Then it uses the SQ algorithm in \Cref{lem:SQ-alg-for-src-poly} to obtain an accurate estimate $\wh{P_{x,w}}(z_0,t_0)$ of $P_{x,w}(z_0,t_0)$ for each $w$ within an additive factor of $\pm 0.1 \cdot \exp(-C_\rho n^{1/5} \log^5 n)$, and outputs the $x'$ such that $\wh{P_{x,w}}(z_0,t_0)$ and $P_{x',w}(z_0,t_0)$ are $\pm 0.5 \cdot \exp(-C_\rho n^{1/5} \log^5 n)$-close to each other for every $w, z_0, t_0$.
The correctness follows immediately from \Cref{lem:distinguish-src-poly}, because if $x' \ne x$, then by that lemma there is some $(w,z_0,t_0)$ such that by the triangle inequality we have 
\begin{align*}
  \abs{\wh{P_{x,w}}(z_0,t_0) - P_{x',w}(z_0,t_0)}
  &\ge \abs{P_{x,w}(z_0,t_0) - P_{x',w}(z_0,t_0)} - \abs{\wh{P_{x,w}}(z_0,t_0) - P_{x,w}(z_0,t_0)} \\
  &\ge 0.9 \cdot \exp(-C_\rho n^{1/5} \log^5 n) . 
\end{align*}
\subsection{Proof of \Cref{lem:distinguish-src-poly}}

In this subsection we prove \Cref{lem:distinguish-src-poly}.
We first recall the following result from \cite{BEK99}.

\begin{lemma}[Theorem 5.1 in \cite{BEK99}] \label{lem:BEK}
  There are constants $c_1, c_2 > 0$ such that for every analytic function $f$ on the open unit disc $\{z: \abs{z} < 1\}$ with $\abs{f(z)} < \frac{1}{1-\abs{z}}$ and every $a \in (0,1]$, we have
    \[
      \abs{f(0)}^{\frac{c_1}{a}}
      \le \exp(c_2/a) \sup_{t \in [1-a,1-\frac{3}{4}a]} \abs{f(t)} .
    \]
\end{lemma}
Note that polynomials with coefficients bounded by $1$ are clearly analytic and satisfy the condition that $\abs{f(z)} < \frac{1}{1-\abs{z}}$ on the open unit disc $\{z: \abs{z} < 1\}$.

We note that in the actual statement in \cite[Theorem~5.1]{BEK99}, the interval containing $t$ is $[1-a,1]$.
However, a close inspection of the proof reveals that the interval can be restricted to be $[1-a, 1-\frac{3}{4}a]$.
Specifically, their Theorem~5.1 is based on their Corollary~5.3, which in turn is based on their Corollary~5.2, where the interval is taken to be $[1-a, 1-a+\frac{1}{4}a]$.
A self-contained proof using essentially the same argument can also be found in \cite[Theorem 9]{CDLSS20smoove}.

We further note that the difference between $[1-a,1]$ and $[1-\frac{3}{4}a]$ is crucial in showing that the contribution of the high-degree terms of the relevant polynomial (\Cref{eq:trace-poly}) is negligible.
Had $t$ been $1$, then $\frac{t - (1-\rho)}{\rho} = 1$ and there would have been no exponential decay in the high-degree terms.

\medskip

\Cref{lem:distinguish-src-poly} follows from two cases below.

\paragraph{Case 1:  $x_i \ne x'_i$ for some $0 \le i \le \ell-1$.}
In this case, we consider the $\ell$-bit pattern $w := x[0:\ell-1]$.
Note that $P_{x,w}(0,0) - P_{x',w}(0,0) = \Indicator{x[0:\ell-1] = w} - \Indicator{x'[0:\ell-1] = w} = 1$.
We now apply \Cref{lem:BEK} twice. The first application is to the polynomial $Q_1(z_1) := P_{x,w}(z_1,0) - P_{x',w}(z_1,0)$, which implies that there exists some $z_0 \in [1-\rho,1-\frac{3}{4}\rho]$ such that
\[
  \abs{Q_1(z_0)}
  \ge e^{-c_2/\rho} \abs{Q_1(0)}^{c_1/\rho}
  = e^{-c_2/\rho} \abs{P_{x,w}(0,0) - P_{x',w}(0,0)}^{c_1/\rho}
  = e^{-c_2/\rho} .
\]
We now apply \Cref{lem:BEK} again to the polynomial 
\[
  Q_2(z_2) := \frac{P_{x,w}(z_0,z_2) - P_{x',w}(z_0,z_2)}{\binom{n}{\ell}} .
\]
Note that all coefficients in $Q_2$ have magnitude at most $1$. 
This implies the existence of some $t_0 \in [1-\rho,1-\frac{3}{4}\rho]$ such that 
\begin{align*}
  \abs{P_{x,w}(z_0,t_0) - P_{x',w}(z_0,t_0)}
  = \binom{n}{\ell} \abs{Q_2(t_0)}
  &\ge \binom{n}{\ell} e^{-c_2/\rho} \abs{Q_2(0)}^{c_1/\rho}
  = \binom{n}{\ell} e^{-c_2/\rho} \left({\frac {\abs{Q_1(z_0)}}{{n \choose \ell}}}\right)^{c_1/\rho} \\
  &\ge \frac {e^{-\frac{c_2}{\rho} - \frac{c_1 c_2}{\rho^2}}}{{n \choose \ell}^{\frac{c_1}{\rho} - 1}}
  \ge e^{-\Omega_\rho(\ell \log n)} = e^{-\Omega_\rho(n^{1/5} \log n)} ,
\end{align*}
where the last inequality used $\binom{n}{\ell} \ge (n/\ell)^\ell$, and the last equality follows from our choice of $\ell = 2n^{1/5}$.
To conclude, there exists some $(z_0,t_0) \in [1-\rho,1-\frac{3}{4}\rho]^2$ such that 
$\abs{P_{x,w}(z_0,t_0) - P_{x',w}(z_0,t_0)} \ge \Omega_\rho({n \choose \ell}^{-c_1/\rho})$.

\paragraph{Case 2:  $x_i = x'_i$ for all $0 \le i \le \ell-1$.}
%
For this case, \cite[Corollary~6.1]{Chase20} (with the interval $[1-2\rho,1]$ replaced with $[1-\rho,1-\frac{3}{4}\rho]$) can be restated, using \Cref{lem:BEK} in a similar fashion as Case 1, as follows: 

\begin{lemma}[Corollary 6.1 in \cite{Chase20}, slightly rephrased and refined]
  For every $\rho > 0$, there exists a constant $C_\rho$ such that the following holds.
  Let $\ell= 2n^{1/5}$.
  For every distinct $x, x' \in \zo^n$ where $x_i = x'_i$ for every $0 \le i < \ell-1$, there exists a pattern $w \in \zo^\ell$, a $z_0 = e^{i\theta}$ for some $\theta \in [-n^{-2/5},n^{-2/5}]$ and a $t_0 \in [1-\rho,1-\frac{3}{4}\rho]$ such that
  \[
  \abs{ \sum_{0 \le i_1 < \cdots < i_\ell \le n-1} \Bigl(\tprod_{k=1}^\ell \Indicator{x_{i_k} = w_k} - \tprod_{k=1}^\ell \Indicator{x'_{i_k} = w_k}\Bigr) z_0^{i_1}  \cdot t_0^{i_\ell - i_1 - (\ell-1)} }
    \ge \exp\bigl(-C_\rho n^{1/5} \log^5 n\bigr) .
  \]
\end{lemma}
Combining the two cases proves \Cref{lem:distinguish-src-poly}. 

\subsection{Proof of \Cref{lem:SQ-alg-for-src-poly}}

We now prove \Cref{lem:SQ-alg-for-src-poly}.
We first state the following identity relating two multivariate polynomials, each of which is defined in terms of an arbitrary $f: \zo^\ell \to \C$. One of these involves the evaluation of $f$ on the $\ell$-bit (not necessarily consecutive) substrings of the source string $x$, and the other involves the expectation of $f$ evaluated on the $\ell$-bit substrings of a random trace $\by \sim \Del_{\delta}(x)$.
This identity has now appeared in several places such as \cite{CDLSS20smoove,Chase20} (see \cite[Section~5.2]{CDLSS20smoove} for a proof).

\begin{fact} \label{fact:mobius}
  For every $f\colon\zo^\ell \to \C$, $x \in \zo^n$, $\rho \in [0,1]$, and  $z \in \C^\ell$,
  \begin{align*}
    \MoveEqLeft
    \rho^\ell \sum_{0 \le i_1 < \cdots < i_\ell \le n-1} f(x_{i_1}, \ldots, x_{i_\ell}) \bigl((1-\rho) + \rho z_1\bigr)^{i_1} \prod_{k=2}^\ell \bigl((1-\rho) + \rho z_k\bigr)^{i_k - i_{k-1} - 1} \\
    &= \sum_{0 \le j_1 < \cdots < j_\ell \le n-1} \Ex_{\by \sim \Del_{1-\rho}(x)}\bigl[f(\by_{j_1}, \ldots, \by_{j_\ell})\bigr]  z_1^{j_1} \prod_{k=2}^\ell z_k^{j_k - j_{k-1} - 1} .
  \end{align*}
\end{fact}

Letting $f(u_1, \ldots, u_\ell)$ be the indicator function $\Indicator{u = w}$ for some pattern $w \in \zo^\ell$, then performing a simple change of variable $z_i \mapsto \frac{z_i - (1-\rho)}{\rho}$, and then identifying the variables $z_3, \ldots, z_\ell$ with the variable $z_2$, we obtain the following corollary.
\begin{corollary} \label{cor:src-trace-id}
  For every $\rho \in (0,1]$, $x \in \zo^n$, $w \in \zo^\ell$, and $(z_1,z_2) \in \C^2$,
  \begin{align} 
    \MoveEqLeft
    P_{x,w}(z_1,z_2) \nonumber \\
    &= \rho^{-\ell} \sum_{0 \le j_1 < \cdots < j_\ell \le n-1} \Ex_{\by \sim \Del_{1-\rho}(x)} \Bigl[\tprod_{k=1}^\ell \Indicator{\by_{j_k} = w_k} \Bigr] \Bigl(\frac{z_1 - (1-\rho)}{\rho}\Bigr)^{j_1}  \Bigl(\frac{z_2 - (1-\rho)}{\rho}\Bigr)^{j_\ell - j_1 - (\ell-1)} . \label{eq:trace-poly}
  \end{align}
\end{corollary}

Let $Q(z_1,z_2)$ be the bivariate polynomial on the right hand side of 
\Cref{eq:trace-poly}.
Observe that for every fixed $z_1$, viewing $Q(z_1,z_2)$ as a univariate polynomial in $z_2$, its $z_2$-coefficient of degree $d$ (a univariate polynomial in $z_1$) can be estimated using $d$-local SQs.
We will first prove that $Q$, as a univariate polynomial in the second variable $z_2$, is close to its low-degree truncation $Q_{\le d}$ (for a suitable choice of $d$), defined by
\begin{align} \label{eq:low-degree-trunc}
  Q_{\le d}(z_1,z_2)
  := \rho^{-\ell} \sum_{\substack{0 \le j_1 < \cdots < j_\ell \le n-1: \\ j_{\ell}-j_1 - (\ell-1) \le d}} \Ex_{\by}\Bigl[\prod_{k=1}^\ell \Indicator{\by_{j_k} = w_k} \Bigr] \Bigl(\frac{z_1 - (1-\rho)}{\rho}\Bigr)^{j_1}   \Bigl(\frac{z_2 - (1-\rho)}{\rho}\Bigr)^{j_\ell - j_1 - (\ell-1)} ,
\end{align}
when both $z_1,z_2$ belong to the domain in \Cref{lem:distinguish-src-poly}.

\begin{claim} \label{claim:SQ-1-term}
  Let $C''_\rho$ be a constant, and $d_0 \ge C''_\rho (\ell + n^{1/5}) + 2\log n$.
  For every $z \in \{ e^{i\theta}: \abs{\theta} \le n^{-2/5} \} \cup [1-\rho,1-\frac{3}{4}\rho]$ and $t \in [1-\rho,1-\frac{3}{4}\rho]$, we have $\abs{Q_{\le d_0}(z,t) - Q(z,t)} \le 4 \cdot 2^{-d_0/2}$.
\end{claim}

\begin{proof}
  It suffices to show that for every $d \ge d_0$, the homogeneous degree-$d$ (in the variable $t$) term of $Q$, that is,
  \begin{align} \label{eq:degree-d-term}
    \rho^{-\ell} \sum_{\substack{0 \le j_1 < \cdots < j_\ell \le n-1 \\ 0 \le j_{\ell}-j_1 - (\ell-1) = d}} \Ex\Bigl[\prod_{k=1}^\ell \Indicator{\by_{j_k} = w_k} \Bigr] \Bigl(\frac{z - (1-\rho)}{\rho}\Bigr)^{j_1}   \Bigl(\frac{t - (1-\rho)}{\rho}\Bigr)^{j_\ell - j_1 - (\ell-1)} ,
  \end{align}
  is bounded by $2^{-d/2}$, as then we have $\abs{Q(z,t) - Q_{\le d_0}(z,t)} \le \sum_{d > d_0} 2^{-d/2} = 4 \cdot 2^{-d_0/2}$, as desired.

  We now bound \Cref{eq:degree-d-term} as follows.
  First, the expectation in each term of the summation can be bounded by $1$.
  Second, writing $z$ as $e^{i\theta}$ for some $\abs{\theta} \le n^{-2/5}$, and using $\abs{\cos\theta} \ge 1-\theta^2/2$, we have
  \begin{align*}
    \abs{z - (1-\rho)}^2
    &= \bigl(\cos\theta - (1-\rho)\bigr)^2 + \sin^2\theta
    = 1 - 2(1-\rho)\cos\theta + (1-\rho)^2 \\
    &= 2(1-\rho)(1 - \cos\theta) + \rho^2
    \le (1-\rho) \theta^2 + \rho^2 .
  \end{align*}
  Using $\abs{\theta} \le n^{-2/5}$ and $j_1 \le n$, 
when $z=e^{i\theta}$ for some $\abs{\theta} \le n^{-2/5}$
we have that
\begin{equation}
\label{eq:zbound}
  \abs[\Big]{\frac{z - (1-\rho)}{\rho}}^{j_1}
    \le \biggl(1 + (1-\rho)\Bigl(\frac{\theta}{\rho}\Bigr)^2\biggr)^{j_1/2}
    \le e^{C'_\rho n^{1/5}}
\end{equation}
  for some constant $C'_\rho$.
And when $z \in [1-\rho,1-{\frac 3 4} \rho]$ we have $0 \leq \frac{z - (1-\rho)}{\rho} \le 1/4$ and so \Cref{eq:zbound} is again satisfied (with room to spare).
Similarly, for $t \in[1-\rho, 1-\frac{3}{4}\rho]$ we have $0 \leq \frac{t - (1-\rho)}{\rho} \le 1/4$, and so $\abs{\left(\frac{t - (1-\rho)}{\rho}\right)^d} \le 4^{-d}$.

  Finally, the number of indices $0 \le j_1 < \cdots < j_\ell \le n-1$ with $j_\ell - j_1 - (\ell-1) = d$ is at most $n \cdot \binom{(\ell-2) + d}{\ell-2} \le n \cdot 2^{d + (\ell-2)}$.
  So the degree-$d$ term \eqref{eq:degree-d-term} can be bounded by 
  \[
    \rho^{-\ell} \cdot n \cdot 2^{d + \ell-2} \cdot e^{C'_\rho n^{1/5}} \cdot 4^{-d}
    \le n \cdot (2/\rho)^{\ell} \cdot e^{C'_\rho n^{1/5}} \cdot 2^{-d},
  \]
  which is at most $2^{-d/2}$ whenever $d \ge C''_\rho (\ell + n^{1/5}) + 2\log n$, for some constant $C''_\rho$.
\end{proof}

We now describe our local SQ algorithm to approximate the low-degree polynomial $Q_{\le d}(z,t)$, for any $(z,t) \in \{e^{i\theta}: \abs{\theta} \le n^{-2/5}\} \cup [1-\rho,1-\frac{3}{4}\rho] \times [1-\rho, 1-\frac{3}{4}\rho]$.
Set $d_0 := C''_\rho n^{1/5} \log^5 n \ge  C''_\rho (\ell + n^{1/5}) + 2\log n$.
Our $d_0$-local algorithm makes the following $d_0$-local queries:
\[
  \Ex_{\by \sim \Del_{1-\rho}(x)} \Bigl[\Indicator{\by[j: j + d_0 - 1] = u} \Bigr] \text{ for every $u \in \zo^{d_0}$ and $j \in \{0, \ldots, n-1\}$} .
\]
Let $\wh{p_{u,j}}$ be the estimate of $\E[\Indicator{\by[j: j + d_0 - 1] = u}]$ that is received as a response to the query.
For every fixed tuple $0 \le j_1 < \cdots <j_\ell \le n-1$ such that $j_\ell - j_1 - (\ell-1) \le d_0$, using the identity 
\[
  \E\Bigl[\prod_{k=1}^\ell \Indicator{\by_{j_k} = w_k} \Bigr]
  = \sum_{u \in \zo^{d_0}: \forall k \in [\ell]: u_{j_k - j_1 + 1} = w_k } \E\Bigl[ \Indicator{\by[j_1: j_1+d_0-1] = u} \Bigr],
\]
which is a sum of $2^{d_0 - \ell}$ terms,
the algorithm computes the estimate $\wh{p_{u,j_1, \ldots, j_\ell}}$ of $\E\left[\prod_{k=1}^\ell \Indicator{\by_{j_k} = w_k}\right]$ (using the estimates $\wh{p_{u,j_1}}$ of $\E[\Indicator{\by[j_1:j_1+d_0-1]=u}]$) by
\[
  \wh{p_{w,j_1,\ldots,j_\ell}}
  := \sum_{u \in \zo^{d_0}: \forall k \in [\ell]: u_{j_k-j_1+1} = w_k } \wh{p_{u,j_1}},
\]
for each $w \in \zo^\ell$ and tuple of indices $0 \le j_1 < \cdots < j_\ell \le n-1$ such that $j_\ell - j_1 - (\ell-1) \le d_0$.
If the tolerance for each query is $\tau_0$, then the error of each estimate $\wh{p_{w,j_1,\ldots,j_\ell}}$ is $\pm 2^{d_0 - \ell} \cdot \tau_0$.
Finally, the algorithm computes the estimate
$\wh{Q_{\le d_0}}(z,t)$ of $Q_{\le d_0}(z,t)$ using \Cref{eq:low-degree-trunc}, as
\[
  \wh{Q_{\le d_0}}(z,t)
  := \rho^{-\ell} \sum_{\substack{0 \le j_1 < \cdots < j_\ell \le n-1: \\ j_{\ell}-j_1 - (\ell-1) \le d_0}} \wh{p_{w,j_1,\ldots,j_\ell}} \Bigl(\frac{z - (1-\rho)}{\rho}\Bigr)^{j_1}  \Bigl(\frac{t - (1-\rho)}{\rho}\Bigr)^{j_\ell - j_1 - (\ell-1)} .
\]
There are at most $n \cdot \binom{d_0 + (\ell-1)}{\ell-1} \le n \cdot 2^{d_0 + (\ell-1)}$ such tuples.
So the total error is $n \cdot 2^{d_0 + (\ell-1)} \cdot 2^{d_0 - \ell} \cdot \tau_0 \le n \cdot 2^{2d_0} \cdot \tau_0$.

By \Cref{claim:SQ-1-term}, we have that for every $(z,t)$ in the domain specified in \Cref{lem:SQ-alg-for-src-poly}
\begin{align*}
  \abs{\wh{Q_{\le d_0}}(z,t) - P_{x,w}(z,t)}
  &= \abs{\wh{Q_{\le d_0}}(z,t) - Q(z,t)} \\
  &\le \abs{\wh{Q_{\le d_0}}(z,t) - Q_{\le d_0}(z,t)} + \abs{Q_{\le d_0}(z,t) - Q(z,t)} \\
  &\le n \cdot 2^{2d_0} \cdot \tau_0 + 4 \cdot 2^{-d_0/2} \\
  &= 2^{2 C''_\rho n^{1/5} \log^5 n} \tau_0 + \exp(-C''_\rho n^{1/5} \log^5 n) .
\end{align*}
Setting the tolerance parameter $\tau_0$ to be $\exp(-C_\rho n^{1/5} \log^5 n)$ proves \Cref{lem:SQ-alg-for-src-poly}. 


\section{Average-case lower bounds} \label{sec:avg-lower}

\begin{theorem} [Average-case lower bound]  \label{thm:avglower}
Fix any constant deletion rate $0<\delta<1$.  Any $\ell$-local SQ algorithm for average-case trace reconstruction must have tolerance $\tau_0 \leq O(\ell/\sqrt{n}).$
%
%
%
\end{theorem}

Let $x$ be an arbitrary fixed string in $\{0,1\}^n$ and let $x'$ be the string
  obtained from $x$ by flipping the bit $x_{n/2}$ in the middle.
Let $q:\{0,1\}^n\rightarrow [-1,1]$ be any $\ell$-junta query (which is not 
  necessarily $\ell$-local), i.e., there are $0\le i_1<\ldots<i_\ell< n$
  such that $q(x)= q'(x_{i_1},\ldots,x_{i_\ell})$ for some $q':\{0,1\}^\ell\rightarrow
  [-1,1]$.
We will prove the following claim:

\begin{claim}\label{mainclaim}
Let $P_q := \Ex_{\by \sim \Del_\delta(x)}[q(\by)]$ and $P_q' := \Ex_{\by \sim \Del_\delta(x')}[q(\by)].$ Then $|P_q-P_q'|\le O(\ell/\sqrt{n})$.
\end{claim}
\begin{proof}
Let $\bR$ be a $\rho$-biased random draw of a subset of $[0:n-1]$
  with $\bR=\{\br_0,\br_1,\ldots,\br_{\boldm-1}\}$ for some $\boldm\le n$.
Given that the only difference between $x$ and $x'$ is the middle bit,
  we have 
$$
\big| P_q-P_q'\big|\le 2\cdot \Pr_{\bR} \left[
\text{$\br_{i_j}=n/2$ for some $j\in [\ell]$}
\right]\le 2\sum_{j\in [\ell]} \Pr_{\bR} \left[
\text{$\br_{i_j}=n/2$}
\right].
$$
Since $\delta\in (0,1)$ is a constant, 
  $\Pr_{\bR}[\br_{i}=n/2]\le O(1/\sqrt{n})$ for any $i$, from which the claim follows. 
\end{proof}

We now prove \Cref{thm:avglower}:

\begin{proof}(of \Cref{thm:avglower})
Indeed we will show that any SQ algorithm for average-case trace reconstruction that uses 
  $\ell$-junta queries with tolerance $\tau$ must satisfy $\tau\le O(\ell/\sqrt{n})$.
To see this, consider any SQ algorithm for trace reconstruction that uses
  $\ell$-junta queries with tolerance $\tau$ that is larger than the 
  $O(\ell/\sqrt{n})$ in \Cref{mainclaim}.
It follows from \Cref{mainclaim} that, for any string $x\in \{0,1\}^n$,
  such an algorithm cannot distinguish between $x$ and $x'$.
As a result, such an algorithm fails to reconstruct $\bx\sim\{0,1\}^n$ 
  with probability at least $1/2$.
\end{proof}


\section{Average-case upper bounds} \label{sec:avg-upper}

\begin{theorem} [Average-case upper bound]  \label{thm:avgupper}
Fix any constant deletion rate $0<\delta<1$.  There is an SQ algorithm for average-case trace reconstruction that uses $\ell=O(\log n)$-local queries with tolerance $\tau=1/\poly(n)$.
\end{theorem}

We will prove Theorem~\ref{thm:avgupper} by showing that the algorithm in \cite{CDLSS22} can be simulated with local SQ queries.
To do so, we will need to recall the smoothed analysis model. 
\begin{definition}
Let $x^{\mathsf{worst}}$ be an unknown and arbitrary string in $\{0,1\}^n$ and $0<\sigma<1$ be a ``smoothening parameter.'' Let $\bx$ be generated by flipping every bit of $x^{\mathsf{worst}}$ independently with probability $\sigma$. 

For parameters $\eta, \tau>0$, a \emph{$(T,\eta,\tau)$-trace reconstruction algorithm in the smoothed analysis model (with smoothening parameter $\sigma$)} has the following guarantee: With probability at least $1-\eta$ (over  the random generation of $\bx$ from $x^{\mathsf{worst}}$), it is the case that the algorithm, given access to independent
traces drawn from $\Del_\delta(\bx)$, outputs the string $\bx$ with probability at least $1-\tau$ (over the random
traces drawn from $\Del_\delta(\bx)$). The time complexity as well as the number of traces is bounded by $T$. 
\end{definition}

Observe that the average case trace reconstruction setting corresponds to the smoothed analysis setting with $\sigma =1/2$ and $ x^{\mathsf{worst}}$ set to the all zeros string (though any fixed  choice of $ x^{\mathsf{worst}}$ works equally well). 

\cite{CDLSS22} gave a polynomial-time algorithm for trace reconstruction in the smoothed analysis setting.
Taking $\sigma=1/2$, the main result of \cite{CDLSS22} gives the following:
\begin{theorem} (Theorem~1 in \cite{CDLSS22}) 
There is an algorithm for trace reconstruction which for any $\eta, \tau$ and $\delta>0$, 
has the following guarantee: With probability $1-\eta$ over $\bx$ drawn uniformly at random from $\{0,1\}^n$, it is the case that the algorithm, given access to independent traces drawn from $\Del_\delta(\bx)$, outputs the string $\bx$ with probability at least $1-\tau$ (over the random
traces drawn from $\Del_\delta(\bx)$). Its running time and sample complexity are upper bounded by
\[
T= \bigg(\frac{n}{\eta} \bigg)^{O\big(\frac{1}{(1-\delta)} \cdot \log\big(\frac{2}{(1-\delta)}\big)\big)}. 
\]
\end{theorem}
 We begin with a short description of the algorithm in \cite{CDLSS22} (page 27 of the Arxiv version of \cite{CDLSS22}), giving only the level of detail necessary for the current paper. 
We set the following parameters: 
\begin{eqnarray}
k &=&  O(\log (n/\eta)), \quad \kappa = \bigg( \frac{1}{n} \cdot \bigg( \frac{1-\delta}{2}\bigg)^k \bigg)^{O(1/(1-\delta))}, \quad 
 \quad \theta = (1-\delta)^2/2, \quad \quad \nonumber \\ d &=& \frac{C}{\theta} \bigg(\ln n +  k\ln \frac{C}{\theta} \bigg), \quad \quad
\Delta = \frac{\kappa}{2d^2 \cdot n \cdot \binom{d+k-2}{k-2}}.  \label{eq:parameter-setting-avg}
\end{eqnarray} 
Set $L$ to be the largest integer such that $\delta + L \cdot \delta \le (1+\delta)/2$. 

Given two strings $x \in \{0,1\}^n$ and $w \in \{0,1\}^k$, \cite{CDLSS22} define a  univariate polynomial $\mathsf{SW}_{x,w}(\cdot)$.  The precise formal definition of this polynomial is not important for us; rather, the following relation (Equation~6 in the Arxiv version of \cite{CDLSS22}) is sufficient for our purposes:
\begin{equation}~\label{eq:rel-SW}
\Ex_{\by \sim \Del_{\delta'}(x)} [\#(w,\by)] =(1-\delta')^k \cdot  \mathsf{SW}_{x,w}(\delta'), 
\end{equation}
 where $\#(w,\by)$ is the number of times $w$ appears a subword of $\by$. We remark to the reader that
 $\#(w,\by)$ is a sum of $k$-local query functions (we will elaborate on this shortly).  The algorithm in \cite{CDLSS22} proceeds as follows:
 \begin{enumerate}
 \item Define set $S: = \{\delta, \delta + \Delta, \delta +2 \Delta, \ldots, \delta + L \Delta\}$.
 \item For every $w \in \{0,1\}^k$ and $\delta' \in S$, the algorithm computes $\pm \kappa$-accurate estimates of $\mathsf{SW}_{x,w}(\delta')$, using \Cref{eq:rel-SW}. 
 \item With these estimates of $\mathsf{SW}_{x,w}(\delta')$ (for $\delta' \in S$ and $w \in \{0,1\}^k$), the algorithm runs a linear program followed by a greedy algorithm to reconstruct the original string. 
 \end{enumerate}
 In particular, excluding the part of Step~(2) in which the estimates of $\mathsf{SW}_{x,w}(\delta')$ are computed, the rest of the reconstruction algorithm is deterministic and does not use the traces. Note that the reason why the \cite{CDLSS22} algorithm does not immediately translate to a local SQ algorithm for us is the following: in our model the permissible statistical queries are with respect to $\by \sim \Del_\delta(x)$, whereas the \cite{CDLSS22} algorithm, as sketched above, uses estimates of probabilities corresponding to statistical queries over $\by \sim \Del_{\delta'}(x)$ for various values of $\delta' > \delta.$

 Thus, to obtain a local SQ algorithm, it suffices to show that the values $\{\mathsf{SW}_{x,w}(\delta')\}_{\delta' \in S,  \ w \in \{0,1\}^k}$ can be estimated using local SQ queries corresponding to $\Del_\delta(x).$ More precisely, we have the following claim whose proof is immediate from the description of the above algorithm and \eqref{eq:rel-SW}.
 \begin{claim}
 For any $\delta' \in S$, to compute $\mathsf{SW}_{x,w}(\delta')$ to error $\kappa$, it is sufficient to estimate the value of $\E_{\by \sim \Del_{\delta'}(x)} [\#(w,\by)]$  up to error $\tau'$, where 
 \begin{equation}~\label{eq:def-tau}
 \tau' := \bigg( \frac{1}{n} \bigg( \frac{1-\delta}{2}\bigg)^k \bigg)^{O(1/(1-\delta))}. 
 \end{equation}
 \end{claim}
 \begin{proof}
 We need to compute $\mathsf{SW}_{x,w}(\delta')$ for $\delta' \in S$ to error $\kappa$. By \eqref{eq:rel-SW}, it suffices to compute $\mathbf{E}_{\by \sim \Del_{\delta'}(x)} [\#(w,\by)]$ to error $\kappa \cdot (1-\delta')^k$; noting that $\delta' \le (1+\delta)/2$, the claim follows. 
 \end{proof}
 
 The main technical lemma of this section is the following.  
 \begin{lemma}~\label{lem:simulate-local-different-delta}
For the parameters defined as above, the following holds: Given the values of all subword queries  of length $\ell$ with tolerance $\tau'/2$ (with $\tau'$ defined in \eqref{eq:def-tau}) corresponding to $\Del_{\delta}(x)$, we can compute $\mathsf{SW}_{x,w}(\delta')$ for all 
 $\delta' \in S$ and $w  \in \{0,1\}^k$ to within error $\pm \kappa$.  Here $\ell$ is defined to be 
 \[
 \ell =  \Theta \bigg( \frac{k}{1-\delta} \cdot \ln \bigg(\frac{2}{1-\delta} \bigg) + 
 \frac{\ln n}{(1-\delta)} 
 \bigg).
 \]
 \end{lemma}
 Before proving this lemma, we observe that \Cref{thm:avgupper} follows immediately from the lemma:
 \begin{proof}(of \Cref{thm:avgupper})
 For any constant $0<\delta<1$ and $\eta = n^{-\Theta(1)}$, by our choice of parameters we have $k =O(\log n)$ (see \eqref{eq:parameter-setting-avg}). With this choice, $\ell = O(\log n)$ and $\tau' = n^{-\Theta(1)}$.
By \Cref{lem:simulate-local-different-delta}, using the values of all subword queries of length $\ell$ (with tolerance $\tau'/2$), we can compute  $\mathsf{SW}_{x,w}(\delta')$ for all 
 $\delta' \in S$ and $w  \in \{0,1\}^k$ to within error $\pm \kappa$. By the guarantee of the algorithm in \cite{CDLSS22}, this suffices to recover $x$. 
  Thus, we get~\Cref{thm:avgupper}. 
 \end{proof}
 
 \subsection{Proof of \Cref{lem:simulate-local-different-delta}}
 
We start with the following observation. 
\begin{fact} \label{fact:simulation-channel}
For $\delta' \ge  \delta$, let $0 \leq \beta_r = (1-\delta')/ (1-\delta) \leq 1$. Then $\Del_{\delta'}(x) = \Del_{\beta_r}(\Del_{\delta}(x))$. In other words, we can simulate traces from the deletion channel $\Del_{\delta'}(\cdot)$ by first getting a trace from $\Del_{\delta}(\cdot)$ and then passing it through the deletion channel $\Del_{\beta_r}$. 
\end{fact}

As stated earlier, we will assume that our original string $x$ is padded with infinitely many $0$-symbols to its right. This means that for any $i$, the $i^{th}$ position of the trace is well-defined. We now consider the process of getting a trace $\by$ from $\Del_{\delta'}(x)$ given a trace $\bz \sim \Del_{\delta}(x)$. We will do this by thinking of $\Del_{\beta_r}(\cdot)$ as a ``selector process". We start 
with the following definition. 
\begin{definition}~\label{def:hyperlocation}
For a parameter $p \in (0,1)$ and  
integers $k >0$ and $\ell \ge 0$, we define the distribution $\Hypernb{p}{k}{\ell}$  as follows: Define an infinite random string $\bw=(\bw_0,\dots)$ in $\{0,1\}^\ast$ where each bit is independently $0$ with probability $p$ and $1$ with probability $(1-p)$. Let $\mathbf{i}_s$ be the location of the $s^{th}$ one in $\bw$. Then a sample from $\Hypernb{p}{k}{\ell}$ 
is given by $(\mathbf{i}_k, \ldots, \mathbf{i}_{k + \ell-1})$. 

Finally, we say that an outcome from $\Hypernb{p}{k}{\ell}$ is \emph{$t$-bounded} if $|\mathbf{i}_{k + \ell-1}- \mathbf{i}_{k }| \le t$. 
\end{definition} 

We note that for any fixed $s$,  the process generating $\mathbf{i}_s$ is {\em memoryless}, in the sense that 
for any fixed $r$ and $s$ (with $r \ge s$), the random variable $\mathbf{i}_r - \mathbf{i}_s$ is distributed as a negative binomial random variable. 

With the above definition, we can now state the following claim:
\begin{claim}~\label{clm:simul-1}
Fix $\delta' \ge \delta$, $k \ge 1$, and $\ell \ge 0$.  Let $\by \sim \Del_{\delta'}(x)$ and $\bz \sim \Del_{\delta}(x)$. For $\beta_r = (1-\delta')/(1-\delta)$, let $(\mathbf{i}_{k}, \ldots, \mathbf{i}_{k +\ell-1}) \sim \Hypernb{\beta_r}{k}{\ell}$.  Then the distribution of $(\by_{k}, \ldots, \by_{k +\ell-1})$ is identical to the distribution of 
$(\bz_{\mathbf{i}_k}, \ldots, \bz_{\mathbf{i}_{k+\ell-1}})$. 
\end{claim}
\begin{proof}
The proof is essentially obvious from~\Cref{fact:simulation-channel} and~\Cref{def:hyperlocation}. In particular, from~\Cref{fact:simulation-channel}, given $\bz \sim \Del_{\delta}(x)$, to get $\by \sim \Del_{\delta'}(x)$, we need to simulate the deletion channel $\Del_{\beta_r}$ on the string $\bz$. By definition of the deletion channel $\Del_{\beta_r}$, the location of the positions $(k, \ldots, k + \ell -1)$ is given by 
 $(\mathbf{i}_k, \ldots, \mathbf{i}_{k + \ell-1})$ sampled from $\Hypernb{p}{k}{\ell}$. This finishes the proof. 
\end{proof}
 
 We next need a  lower bound on the probability that $\Hypernb{p}{k}{\ell}$ is bounded. To obtain this, we first state a tail bound 
 on negative binomial random variables:
 \begin{claim} [\cite{Brown11}] \label{clm:negbin-concentration}
 Let $\mathsf{Negbin}(m, p)$ be a negative binomially distributed random variable  with parameters $m$ and $p$, i.e.~it is the number of trials needed to get $m$ heads from independent coin tosses with heads probability $p$. Then $\mathbf{E}[\mathsf{Negbin}(m, p)] = m/p$ and furthermore, for any $t>1$, 
\[
\Pr\bigl[\mathsf{Negbin}(m, p)>tm/p\bigr] \le \exp \bigg(-\frac{tm (1-1/t)^2}{2} \bigg). 
\] 
 \end{claim}
 From this we can obtain the following claim which lower bounds the probability that $\Hypernb{\beta_r}{k}{\ell}$ is $s$-bounded. 
 \begin{claim}~\label{clm:negbin-2}
 For any $k$, an outcome $(\mathbf{i}_k, \ldots, \mathbf{i}_{k + \ell-1}) \sim \Hypernb{\beta_r}{k}{\ell}$ is $s$-bounded with probability at least $1-\xi$ for $s= t(\ell-1)/\beta_r $, where $\xi = \exp(-t(\ell-1)/8)$ for $t \ge 2$.  
 \end{claim}
 \begin{proof}
The gap $\mathbf{i}_{k + \ell-1} - \mathbf{i}_{k }$ is a negative binomial random variable which is distributed as $\mathsf{Negbin}(\ell-1, \beta_r)$. Thus, by~\Cref{clm:negbin-concentration}, it follows that 
 \[
\Pr\Bigl[ \mathbf{i}_{k + \ell-1} - \mathbf{i}_{k } > \frac{t(\ell-1)}{\beta_r}\Bigr]
\le \exp \bigg( \frac{-t(\ell-1)(1-1/t)^2}{2}\bigg). 
 \]
 For $t \ge 2$, we can simplify the upper bound as 
 \[
\Pr\Bigl[ \mathbf{i}_{k + \ell-1} - \mathbf{i}_{k } > \frac{t(\ell-1)}{\beta_r}\Bigr]
\le \exp \bigg( \frac{-t(\ell-1)}{8}\bigg). 
 \]
 Defining $\xi$ as $\exp(\frac{-t(\ell-1)}{8})$, we get the claim. 
 \end{proof}

 We now state the following technical claim. 
 \begin{claim}~\label{clm:tech-def-claim-2}
 Given the value of all $\ell$-local subword queries  for deletion channel $\Del_{\delta}(x)$ with tolerance $\eta$, we can compute the value of all $\ell'$-local subword queries for $\Del_{\delta'}(x)$ with tolerance $\tau'$ where 
 \[
 \eta = \tau'/2; \quad \ell = C \cdot \bigg( \frac{\ell'}{1-\delta'} \cdot \ln \bigg(\frac{2}{1-\delta} \bigg) + 
 \frac{\ln n}{(1-\delta')} 
 \bigg),
 \]
 for a suitably large constant $C$.
 \end{claim}
 \begin{proof} 
   Fix any $w \in \{0,1\}^{\ell'}$ and consider the quantity
   \[
 p'_{x,k,w} := \Prx_{\by \sim \Del_{\delta'}(x)} [(\by_{k}, \ldots, \by_{k+\ell'-1}) =w]. 
 \] 
 Then, by \Cref{clm:simul-1}, it follows that 
 \begin{equation}~\label{def:pxlw}
 p'_{x,k,w} := \Prx_{\bz \sim \Del_{\delta}(x), (\mathbf{i}_k, \ldots, \mathbf{i}_{k + \ell'-1}) \sim \Hypernb{\beta_r}{k}{\ell'}} [\bz_{\mathbf{i}_1}, \ldots, \bz_{ \mathbf{i}_{k + \ell'-1}} =w], 
 \end{equation}
 where $\beta_r = (1-\delta')/(1-\delta)$. Define the parameter 
 \[
 t = C \cdot \bigg( \frac{1}{1-\delta} \cdot \ln \bigg(\frac{2}{1-\delta} \bigg) + 
 \frac{\ln n}{\ell'(1-\delta)} 
 \bigg), 
 \]
 where the constant $C$ is set  so that 
 \[
 \exp \bigg(\frac{t(\ell'-1)}{8} \bigg) = \frac{\tau'}{2}. 
 \]
 As $\ell=t(\ell'-1)/\beta_r$, it follows 
 from \Cref{clm:negbin-2} that
 \begin{equation}\label{eq:ref-E-bound}
   \Pr[ \mathbf{i}_{\ell'+k-1} - \mathbf{i}_{\ell'} > \ell]  \le \exp \bigg(\frac{t(\ell'-1)}{8} \bigg) =\frac{\tau'}{2}. 
 \end{equation}
 Now, define  $\mathcal{E}$ as the event (over the samples $ (\mathbf{i}_k, \ldots, \mathbf{i}_{k + \ell'-1})$)  that  $|\mathbf{i}_{k + \ell'-1} - \mathbf{i}_k| \le \ell$. 
We now re-express 
 \begin{eqnarray}
 p'_{x,k,w} &=& \Prx_{\bz \sim \Del_{\delta}(x),(\mathbf{i}_k, \ldots, \mathbf{i}_{k + \ell'-1})} \bigl[\bz_{\mathbf{i}_1}, \ldots, \bz_{ \mathbf{i}_{k + \ell'-1}} =w \ \wedge \mathcal{E} \bigr] \nonumber 
 \\ &+& \Prx_{\bz \sim \Del_{\delta}(x), (\mathbf{i}_k, \ldots, \mathbf{i}_{k + \ell' -1})} \bigl[\bz_{\mathbf{i}_1}, \ldots, \bz_{ \mathbf{i}_{k + \ell' -1}} =w \ \wedge \overline{\mathcal{E}} \bigr]. \nonumber
 \end{eqnarray}
 From the bound \eqref{eq:ref-E-bound}, the second term is at most $\tau'/2$ in magnitude and thus,
 \[
\Bigl| p'_{x,k,w}- \Prx_{\bz \sim \Del_{\delta}(x), (\mathbf{i}_k, \ldots, \mathbf{i}_{k + \ell' -1})} \bigl[\bz_{\mathbf{i}_1}, \ldots, \bz_{ \mathbf{i}_{k + \ell' -1}} =w \ \wedge \mathcal{E} \bigr]\Bigr| \le \tau'/2. 
 \]
 Furthermore, for any particular outcome of $(\mathbf{i}_k, \ldots, \mathbf{i}_{k + \ell'-1})$ for which event $\mathcal{E}$ happens, 
 the quantity 
 $$
\Prx_{\bz \sim \Del_{\delta}(x)}  [\bz_{\mathbf{i}_1}, \ldots, \bz_{ \mathbf{i}_{k + \ell'-1}} =w], 
$$
is a $\ell$-local subword query. Since we have the value of all $\ell$-local subword queries up to error $\tau'/2$, it follows that we can compute $p'_{x,k,w}$ to error $\tau'$. 
 \end{proof}
 \begin{proof}(of \Cref{lem:simulate-local-different-delta})
   By \Cref{clm:simul-1}, to compute $\mathsf{SW}_{x,w}(\delta')$ to error $\pm\kappa$, it suffices to compute $\mathbf{E}_{\by \sim \Del_{\delta'}(x)} [\#(w,\by)]$ for every $w \in \{0,1\}^{\ell'}$ up to error $\pm\tau'$ where $\tau'$ is defined in \eqref{eq:def-tau}. Now, by \Cref{clm:tech-def-claim-2}, for any given $\delta' \geq \delta$, to compute $\mathbf{E}_{\by \sim \Del_{\delta'}(x)} [\#(w,\by)]$ to error $\tau$, it suffices to have the value of all $\ell$-local subword queries to error $\tau'/2$ where 
 \[
  \ell = C \cdot \bigg( \frac{\ell'}{1-\delta'} \cdot \ln \bigg(\frac{2}{1-\delta} \bigg) + 
 \frac{\ln n}{(1-\delta')} 
 \bigg). 
 \]
 Since $\delta' \le (1+\delta)/2$, it follows that 
 $$
 \ell \le C \cdot \bigg( \frac{2\ell'}{1-\delta} \cdot \ln \bigg(\frac{2}{1-\delta} \bigg) + 
 \frac{2\ln n}{(1-\delta)} 
 \bigg).
 $$
 Thus, if we have the value of all $k$-local subword queries to error $\tau'/2$, where $k$ is set to 
 \[
 k = \Theta \bigg( \frac{\ell'}{1-\delta} \cdot \ln \bigg(\frac{2}{1-\delta} \bigg) + 
 \frac{\ln n}{(1-\delta)}\bigg), 
 \]
 we can recover $x$. This finishes the proof. 
 \end{proof}
 

\begin{flushleft}
\bibliography{allrefs}{}
\bibliographystyle{alpha}
\end{flushleft}

\end{document}